\documentclass{llncs}

\usepackage{amssymb}
\usepackage{amsmath}
\usepackage{graphicx}

\usepackage{color}
\usepackage[usenames,dvipsnames,svgnames,table]{xcolor}

\usepackage{bbm}

\usepackage[OT4]{fontenc}
\usepackage[utf8]{inputenc}

\usepackage{listings}

\newcommand{\mket}[1]{| #1 \rangle}
\newcommand{\mbra}[1]{\langle #1 |}

\newcommand{\mtr}[1]{\mathrm{Tr}\left( #1 \right)}
\newcommand{\mptr}[2]{\mathrm{Tr_{#2}}\left( #1 \right)}

\newcommand{\nH}{\mathcal{H}}

\newcommand{\imag}{\mathbbm{i}}

\newcommand{\nHTOT}{H_{\mathrm{TOT}}}

\newcommand{\keywords}[1]{\par\addvspace\baselineskip\noindent\keywordname\enspace\ignorespaces#1}

\begin{document}


\mainmatter

\title{Quantum Coherence Measures \\ for Quantum Switch}



\authorrunning{Marek Sawerwain \and Joanna Wi\'sniewska}
\tocauthor{Marek Sawerwain and Joanna Wi\'sniewska}

\author{Marek Sawerwain\inst{1} \and Joanna Wi\'sniewska\inst{2}}
\institute{Institute of Control \& Computation Engineering \\
University of Zielona G\'ora, Licealna 9, Zielona G\'ora 65-417, Poland \\
\email{M.Sawerwain@issi.uz.zgora.pl}
\and
Institute of Information Systems, Faculty of Cybernetics, \\
Military University of Technology,  Gen.~W.~Urbanowicza 2, 00-908 Warsaw, Poland \\
\email{jwisniewska@wat.edu.pl}
}


\maketitle

\begin{abstract}
We suppose that a structure working as a quantum switch will be a significant element of future networks realizing transmissions of quantum information. In this chapter we analyze a process of switch's operating -- especially in systems with a noise presence. The noise is caused by a phenomenon of quantum decoherence, i.e. distorting of quantum states because of an environmental influence, and also by some imperfections of quantum gates' implementation. In the face of mentioned problems, the possibility of tracing the switch's behavior during its operating seems very important. To realize that we propose to utilize a Coherence measure which, as we present in this chapter, is sufficient to describe operating of the quantum switch and to verify correctness of this process. It should be also stressed that the value of Coherence measure may be estimated by a quantum circuit, designed especially for this purpose.
\end{abstract}

\section{Introduction} \label{lbl:sec:introduction:MS:JW:CN:2018}

A field of quantum computing \cite{Wegrzyn2001}, \cite{Wegrzyn2002} called quantum communication \cite{Imre2012}, \cite{VanMeter2014}, \cite{Cariolaro2015}, \cite{Li2016} is a dynamically growing research area of network science \cite{Goscien2017}, \cite{Markowski2017}. Quantum communication deals, inter alia, with processing of quantum information. The concept of information transmission/switching in a quantum channel is one of the basic issues connected with quantum communication. In case of quantum information its transmission and switching are non-trivial problems because of non-cloning theorem \cite{WoottersZurek}, \cite{Jozsa2002b} and decoherence \cite{Gawron2012}.  

A definition of quantum switch, realizing swapping information between channels A and B, was presented in \cite{Ratan2007}. A significant issue is to trace the operating of quantum switch, especially when distortions (caused by an influence of external environment) of quantum information occur. Correctness of the process may be, for example, verified with use of quantum entanglement phenomenon -- some basic information concerning this method was contained in \cite{Brus2011} and \cite{Sawerwain2017}. In this chapter we suggest utilizing the Coherence measure \cite{Baumgratz2014} to evaluate the correctness of switch's operating. Furthermore, following the results described in \cite{Girolami2014} we propose of a quantum circuit that estimates of Coherence measure value for a quantum switch. It should be stressed that estimating the value of Coherence measure allows to trace the behavior of quantum system and clearly points out if it works properly.

The chapter is organized as follows. In Sec.~\ref{lbl:sec:quantum:switch:MS:JW:CN:2018} the basic information concerning the quantum switch was presented. This section contains description of the switch as a Hamiltonian and as a time-dependent unitary operation. There is also an example of distortions modeled with use of Dzyaloshinskii-Moriya interaction \cite{Dzyaloshinskii1958}, \cite{Moriya1960}. Sec.~\ref{lbl:sec:quantum:coherence:MS:JW:CN:2018} contains definitions of Coherence measure, its properties and two examples of popular realizations.

A description of measures in context of the switch are presented in Sec.~\ref{lbl:sec:QC:for:QS:MS:JW:CN:2018}. We showed direct analytical formulas expressing Coherence measure for the switch during its operating. Conducted numerical experiments demonstrate changes in value of Coherence measure for systems with and without noise. We calculated also an analytical formula describing a difference between the values of Coherence measure for systems without presence of noise and with distortions modeled as Dzyaloshinskii-Moriya interaction.

The summary and plans for further work are presented in Sec.~\ref{lbl:sec:conclusions:MS:JW:CN:2018}. A list of references to other works ends this chapter.

\section{Quantum Switch} \label{lbl:sec:quantum:switch:MS:JW:CN:2018}

A quantum switch, analyzed in this chapter, is a system of three qubits denoted as: $A$, $B$, $C$. The states of qubits $A$ and $B$ are unknown: 
\begin{equation}
\mket{A} = \alpha_A \mket{0} + \beta_A \mket{1} , \;\;\;  \mket{B} = \alpha_B \mket{0} + \beta_B \mket{1} .
\end{equation}
The state of $C$ is known and preserves only two possible alternatives: $\mket{C} = \mket{0}$ or $\mket{C} = \mket{1}$.

A main task of the quantum switch is swapping an unknown states between qubits $A$ and $B$. This operation is performed conditionally: if the state of qubit $\mket{C}$ is $\mket{0}$ then there is no action; but if the state of qubit $\mket{C}$ is $\mket{1}$ then the values of quantum states are exchanged between qubits $A$ and $B$:
\begin{equation}
\mket{AB0} \rightarrow \mket{AB0}, \;\;\; \mket{AB1} \rightarrow  \mket{BA1} .
\end{equation}

The quantum switch may be depict as a circuit built of three quantum gates: two CNOT gates and one Toffoli gate. The conditional swapping of quantum states, realized by the mentioned circuit, is presented at Fig.~\ref{lbl:fig:QS:EE:MS:JW:CN:2018}.

\begin{figure}
\begin{center}
\includegraphics[height=3cm]{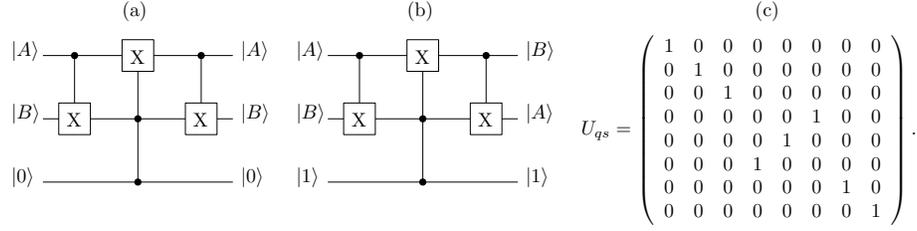}
\end{center}
\caption{An exemplary circuit presenting the action performed by the quantum switch. If the state of third qubit is $\mket{0}$ (case (a)) the switch does not swap the states of first two qubits. When the state of the last qubit is expressed as $\mket{1}$ (case (b)) the circuit swaps states $\mket{A}$ and $\mket{B}$. Case (c) depicts the matrix form of unitary operator which performs action of quantum switch}
\label{lbl:fig:QS:EE:MS:JW:CN:2018}
\end{figure}

Utilizing the circuit, shown at Fig.~\ref{lbl:fig:QS:EE:MS:JW:CN:2018}, we can denote a matrix form (also shown at Fig.~\ref{lbl:fig:QS:EE:MS:JW:CN:2018} -- case (c)) of an operator realizing the operation performed by the quantum switch  describes the complete process of information switching without details concerning particular steps of the process. Calculating a time-dependent unitary operation will allow to present the flow of information through the quantum switch during its operating time. To obtain this unitary operation we need to compute a Hamiltonian describing the dynamics of information flow in a quantum register:

\begin{equation}
H_{qs} = \mket{011}\mbra{101} + \mket{101}\mbra{011}  = \left(
\begin{array}{cccccccc}
 0 & 0 & 0 & 0 & 0 & 0 & 0 & 0 \\ 
 0 & 0 & 0 & 0 & 0 & 0 & 0 & 0 \\ 
 0 & 0 & 0 & 0 & 0 & 0 & 0 & 0 \\ 
 0 & 0 & 0 & 0 & 0 & 1 & 0 & 0 \\ 
 0 & 0 & 0 & 0 & 0 & 0 & 0 & 0 \\ 
 0 & 0 & 0 & 1 & 0 & 0 & 0 & 0 \\ 
 0 & 0 & 0 & 0 & 0 & 0 & 0 & 0 \\ 
 0 & 0 & 0 & 0 & 0 & 0 & 0 & 0 \\ 
\end{array}
\right) .
\end{equation}

The operator $H_{qs}$ is Hermitian, so introducing a variable $t$ we may obtain a~time-dependent unitary operation:
\begin{equation}
U_{qs}(t) = e^{ -\imag t H_{qs}}, 
\hat{U}_{qs}(t) = \left(
\begin{array}{cccccccc}
 1 & 0 & 0 & 0 & 0 & 0 & 0 & 0 \\
 0 & 1 & 0 & 0 & 0 & 0 & 0 & 0 \\
 0 & 0 & 1 & 0 & 0 & 0 & 0 & 0 \\
 0 & 0 & 0 & \imag \cos(t) & 0 & \sin(t) & 0 & 0 \\
 0 & 0 & 0 & 0 & 1 & 0 & 0 & 0 \\
 0 & 0 & 0 & \sin(t) & 0 & \imag \cos(t) & 0 & 0 \\
 0 & 0 & 0 & 0 & 0 & 0 & 1 & 0 \\
 0 & 0 & 0 & 0 & 0 & 0 & 0 & 1 \\
\end{array}
\right).
\label{lbl:eq:unitary:evolution:for:QS:MS:JW:CN:2018}
\end{equation}
The operation $U_{qs}(t)$ causes the change of local phase therefore using the phase-flip gate allows to obtain operator $\hat{U}_{qs}(t)$ which, naturally, does not introduce any change of local phase into the system.

Both, shown above, forms of unitary operation $U_{qs}(t)$ let to trace the process of switch's operating according to time $t$. If the switch acts on a state $\mket{AB1}$ and there is no correction of the local phase then the final state is a pure quantum state:
\begin{equation}
U_{qs}(t) \mket{\Psi_{qs}} = \mket{\Psi^{t}_{qs}} = \left(
\begin{array}{c}
 0 \\
 \alpha _0 \alpha _1 \\
 0 \\
 \cos (t) \alpha _0 \beta _1 - \imag \sin (t) \alpha _1 \beta _0 \\
 0 \\
 \cos (t) \alpha _1 \beta _0 - \imag \sin (t) \alpha _0 \beta _1 \\
 0 \\
 \beta _0 \beta _1 \\
\end{array}
\right) ,
\label{lbl:eqn:QS:qstate:MS:JW:CN2018}
\end{equation}
where $t \in \langle 0, \frac{\pi}{2} \rangle$.

And representation of this operation as a density matrix $\rho$ is also given:
\begin{equation}
\rho(t) = U_qs(t)  \mket{\Psi_{qs}^{t}} \mbra{\Psi_{qs}^{t}}  U^{\dagger}_{qs}(t) .
\end{equation}

In Sec.~\ref{lbl:sec:QC:for:QS:MS:JW:CN:2018} we are going to present the values of Coherence measure with the presence of noise. To do that the following Hamiltonian, describing Dzyaloshinskii-Moriya (DM) interaction \cite{Dzyaloshinskii1958}, \cite{Moriya1960}, will be used: 
\begin{equation}
H_{\mathrm{DM}} = D_z \cdot (\sigma_A^x \otimes \sigma_B^y - \sigma_A^y \otimes \sigma_B^x )
\end{equation}

The notation of operator $\sigma_A^x$ informs us that the Pauli operator X is used on the qubit A and similarly $\sigma_B^y$ means that the Pauli operator Y is used on the qubit B. The symbol $D_z$ represents the vector of process intensity. A Hamiltonian $H_{TOT}$ joins the switch and DM interaction:
\begin{equation}
\nHTOT = t \cdot H_{qs} + D_z \cdot H_{\mathrm{DM}}, \;\;\; U_{qs}^{\mathrm{DM}}(t, D_z) = e^{ -\imag (t \cdot H_{qs} + D_z \cdot H_{\mathrm{DM}}) }.
\end{equation}
where $U_{qs}^{\mathrm{DM}}$ stands for the unitary operation. It should be stressed that for $D_z$ we obtain the Hamiltonian describing only the quantum switch's operating.

A state of quantum register including an influence of DM interaction may be expressed as:
\begin{equation}
U_{qs}^{\mathrm{DM}}(t) \mket{\Psi_{qs}} = \mket{\Psi^{U_{qs}^{\mathrm{DM}}(t)}_{qs}} = \left(
\begin{array}{c}
 0 \\
 \alpha _0 \alpha _1 \\
 0 \\
 \frac{(2 D_z -\imag t) \sinh \left( \gamma \right) \alpha _1 \beta _0}{ \gamma }+\cosh \left( \gamma \right) \alpha _0 \beta _1 \\
 0 \\
 \cosh \left( \gamma \right) \alpha _1 \beta _0-\frac{(2 D_z+\imag t) \sinh \left( \gamma \right) \alpha _0 \beta _1}{\gamma} \\
 0 \\
 \beta _0 \beta _1 \\
\end{array}
\right),
\end{equation}
and  $\gamma = \sqrt{-4 D_z^2 - t^2}$. And also as a density matrix $\rho$: 
\begin{equation}
\rho(t,D_z) = U(t,D_z)  \mket{\Psi_{qs}} \mbra{\Psi_{qs}}  U^{\dagger}(t,D_z) .
\end{equation}

\section{Quantum Coherence Measures} \label{lbl:sec:quantum:coherence:MS:JW:CN:2018}

Introducing the notion of Coherence needs to point the incoherent quantum states. This approach is similar to the measures of quantum entanglement where a set of separable states has to be specified. For a $d$-dimensional Hilbert space $\nH$ we have to define a computational basis, for example the standard computational basis: $\{ \mket{i} \}, \; \mathrm{for} \; i=0,1,2,3, \ldots, d.$.
The Coherence measure is basis-dependent, i.e. we consequently use the standard computational basis in this chapter.

Incoherent states $\mathcal{I}$ and maximally coherent state $\mket{\phi_d}$ are defined as:
\begin{eqnarray}
\mathcal{I} = \{ \delta | \delta = \sum_{d=0}^{d-1} \delta_i \mket{i}\mbra{i} \}, \;\;\; \mket{\phi_d} = \frac{1}{\sqrt{d}} \sum_{i=0}^{d-1} e^{i \phi_i} \mket{i},
\end{eqnarray}
where $\phi_i$ represents a phase of basis element $\mket{i}$. 

Operations performed on states described as density matrices $\rho$ are termed as Incoherent Operations (IO) if for a Complete Positive and Trace Preserving (CPTP) projection $\Lambda(\rho) = \sum_{n} K_n \rho K_n^\dagger$ we observe that $\frac{K_n \delta K_n^\dagger}{\mathrm{Tr}[K_n \delta K_n^\dagger]} \in \mathcal{I}$ for every value of $n$ and $\delta \in \mathcal{I}$. A set of these operations may be denoted as $\Lambda(\delta) \in \mathcal{I}$ for each $\delta \in \mathcal{I}$ and it is called Maximal Incoherent Operations (MIO). Naturally: $IO \subseteq MIO$.

The Coherence measure for $\rho$ is defined by a real-valued function $C(\rho)$. The function $C(\rho)$ has to fulfil the following properties:
\begin{itemize}
\item[(P1)] $C(\rho) \geq 0, \forall \rho $ and $C(\delta) = 0$ if and only if $\delta \in I$,
\item[(P2)] monotonicity -- the coherence measure $C$ cannot increase its value if refers to an operation included in IO or MIO represented by channel $\Lambda$: $C(\Lambda(\rho)) \leq C(\rho)$,
\item[(P3)] strong monotonicity with post-selection -- for any channel $\Lambda \in IO$, with given set of Kraus operators $\{ K_n \}$, the coherence measure $C$ cannot increase its average value under the post-selection:
\begin{displaymath}
\sum_n p_n C(\rho_n) \leq C(\rho), \; \mathrm{where} \; \rho_n = {K_n \rho K_n^\dagger}{\mathrm{Tr}[K_n \rho K_n^\dagger]},
\end{displaymath}
\item[(P4)] convexity -- the value of coherence measure $C$ cannot be increased by mixing quantum states: $C( \sum_n p_n \rho_n ) \leq \sum_n p_n C( \rho_n )$.
\end{itemize}

Generally, if the measure $C$ fulfils the condition P1 and P2 or P3 then it is a monotone function and in this case the measure is very useful. Nowadays, the most widely use Coherence measures are: $l_1$-norm coherence ($C_{l_1}$) and relative entropy of coherence ($C_{\mathrm{re}}$). These both measures fulfill the above-mentioned properties. 

A significant advantage of both specified measures is their relative simplicity. The $l_1$-norm coherence is presented as a sum of all off-diagonal elements of density matrix and the relative entropy of coherence measure is defined as entropy difference:
\begin{equation}
C_{l_1} ( \rho )  = \sum_{\substack{i,j \\ i \neq j }} | \rho_{i,j} |, \;\;\; C_{\mathrm{re}} ( \rho )  = S(\rho_{\mathrm{diag}}) - S(\rho) .
\label{lbl:eq:norm1:cocherence:measure:MS:JW:CN:2018}
\end{equation}
where $S( \cdot )$ stands for von Neumann entropy and $\rho_{\mathrm{diag}}$ denotes a density matrix without the off-diagonal elements.

\section{Coherence Measures for Quantum Switch} \label{lbl:sec:QC:for:QS:MS:JW:CN:2018}

The Coherence measures will be utilized to assess the correctness of quantum switch operating. To do it we need to calculate values of these measures for two situations: when analyzed quantum state contains noise and without it.

If a quantum state is free from noise, we examine states of two first qubits $A$ and $B$. After a partial trace operation, which aim is to remove a state of controlling qubit $C$, we obtain a reduced density matrix:
\begin{equation}
\mptr{\rho_{ABC}}{C} = \rho_{AB},
\end{equation}
where the states $\mket{A}$ and $\mket{B}$ are unknown. Utilizing measure (\ref{lbl:eq:norm1:cocherence:measure:MS:JW:CN:2018}) and carrying out some necessary transformations allows to calculate the following formula:
\begin{multline}
C_{l1}( \rho_{AB}) = 2 \Biggl( \left| \alpha_0 \alpha_1 \beta_0 \beta_1 \right| + \left| \alpha _0 \alpha _1 \left( \cos(t) \alpha_0 \beta_1 -  \imag \sin(t) \alpha_1 \beta_0\right) \right| \\
 + \left| \beta_0 \beta_1 \left( \cos(t) \alpha_0 \beta_1 - \imag \sin(t) \alpha_1 \beta_0 \right) \right| + \left( \left| \alpha_0 \alpha_1 \right| + \left| \beta_0 \beta_1 \right| \right. \\
 \left.  + \left| \cos(t) \alpha_0 \beta_1 - \imag \sin(t) \alpha_1 \beta_0  \right| \right) \left| \cos(t) \alpha_1 \beta_0 - \imag \sin(t) \alpha_0 \beta_1 \right| \Biggr) .
\label{lbl:eq:cocherence:for:AB:MS:JW:CN:2018}
\end{multline}

A value of Coherence measure $C_{l_1}$ depends on time and states of qubits. However, if the second qubit $\mket{B}=\mket{0}$, we can use (\ref{lbl:eq:cocherence:for:AB:MS:JW:CN:2018}) and compute:
\begin{gather}
C_{l1}( \rho_{A0}) = 2 \left| \sin (t) \alpha_0 \beta_0 \right| +2 \left| \cos (t) \alpha_0 \beta_0\right| +2 \left| \cos (t) \sin (t) \beta_0^2 \right|
\end{gather}
As we can see the value of this measure depends only on time elapsing during the switch's operating. Basing on the fact that the quantum state is normalized, we can directly indicate the constraint concerning the value of $C_{l_1}$ measure for the state $\mket{A0}$:
\begin{equation}
C_{l1}( \rho_{A0})  <  \left ( \varepsilon + (\left| \sin(t)\right| +\left| \cos(t)\right| +\left| \cos(t)\sin(t)\right|) \right),
\end{equation}
where the state $\mket{A}$ is unknown and the constant $\varepsilon \in \mathbb{R}$. 

Generally, the changes of $C_{l_1}$ value allow to trace the work of switch because at the beginning of operating, when $t=0$, the value of $C_{l_1}$ is minimal, then in the middle of the process, that is when $t=\pi/4$, the value of $C_{l_1}$ is maximal. When the switch finishes operating in a moment $t=\pi/2$ the value of $C_{l_1}$ is the same as in the moment $t=0$. Fig.~\ref{lbl:fig:coherence:for:A0:and:AB:states:MS:JW:CN:2018} depicts the changes of $C_{l_1}$ measure for quantum states:
\begin{equation}
\mket{A} = \frac{1}{\sqrt{10}} \mket{0} + \frac{\sqrt{9}}{\sqrt{10}}\mket{1}, \mket{B} = \frac{\sqrt{9}}{\sqrt{10}} \mket{0} + \frac{1}{\sqrt{10}} \mket{1}.
\label{lbl:eq:given:state:MS:JW:CN:2018}
\end{equation}
and, more generally, for the states described as:
\begin{equation}
\mket{A} = \sin(a) \mket{0} + \cos(a) \mket{1}, \;\;\; \mket{B} = \sin(\frac{\pi}{2} - a) \mket{0} + \cos(\frac{\pi}{2} - a) \mket{1}
\label{lbl:eq:state:for:a:paramter:MS:JW:CN:2018}
\end{equation}
where $a \in \langle 0, \frac{\pi}{2} \rangle$.

\begin{remark}
If the time of switch's operating is extended, for example, to $t=\pi$ and the value of the parameter $a$ is increased to $\pi$, we will observe a periodic character in changes of $C_{l_1}$ value.
\end{remark}

\begin{figure}
\begin{center}
\begin{tabular}{ccc}
(a) & & (b) \\
\includegraphics[height=3.0cm]{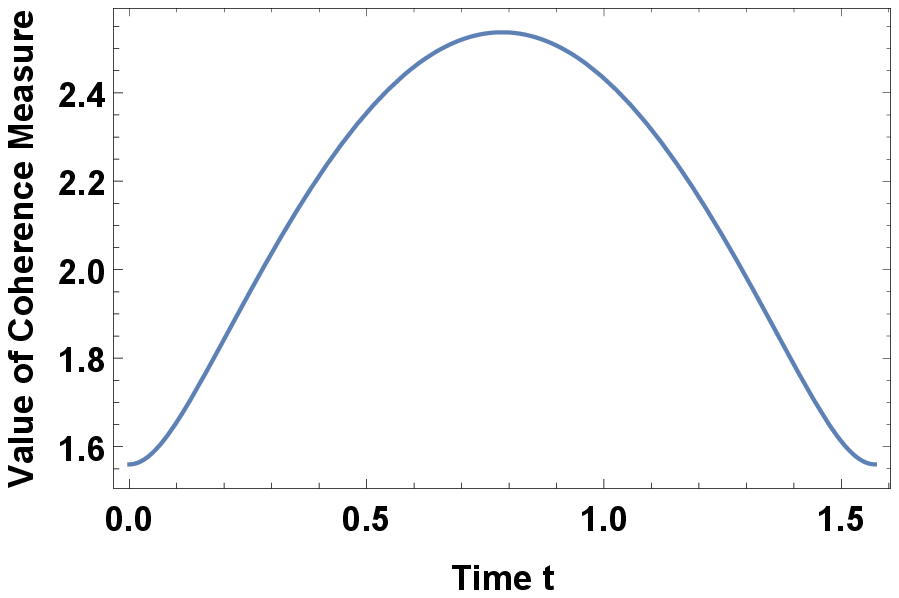} & & \includegraphics[height=3.0cm]{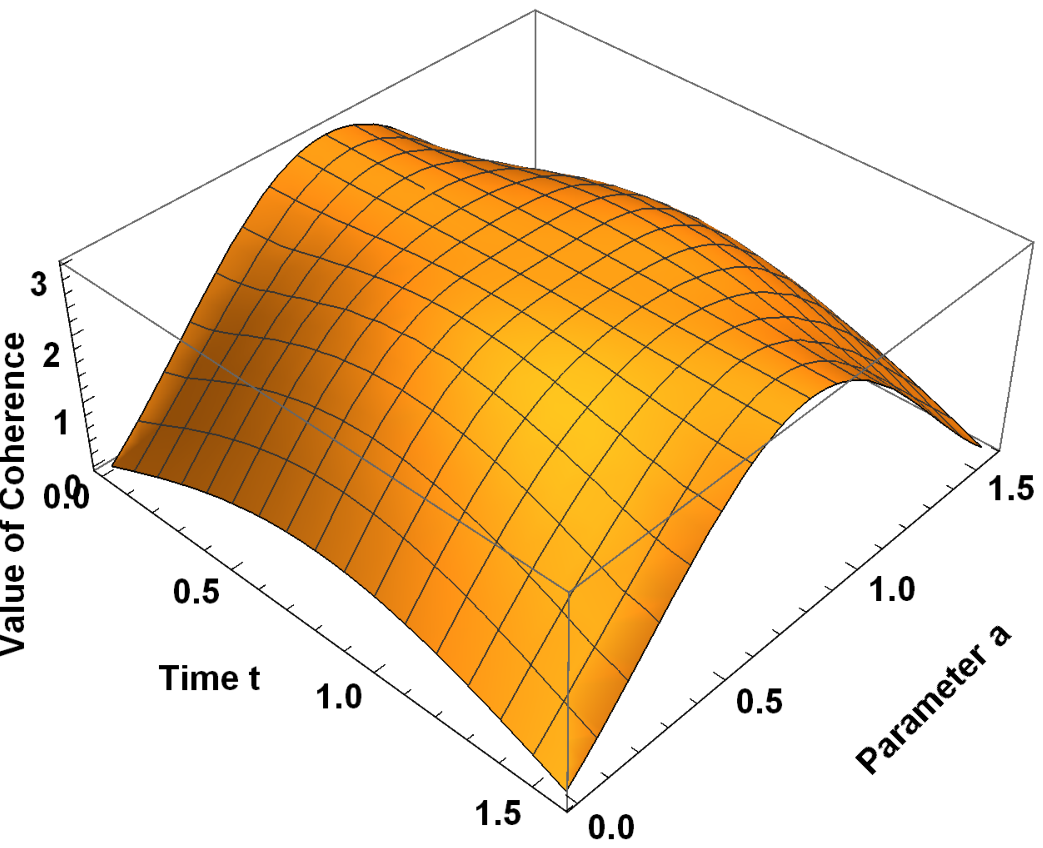} \\
(c) & & (d) \\
\includegraphics[height=3.0cm]{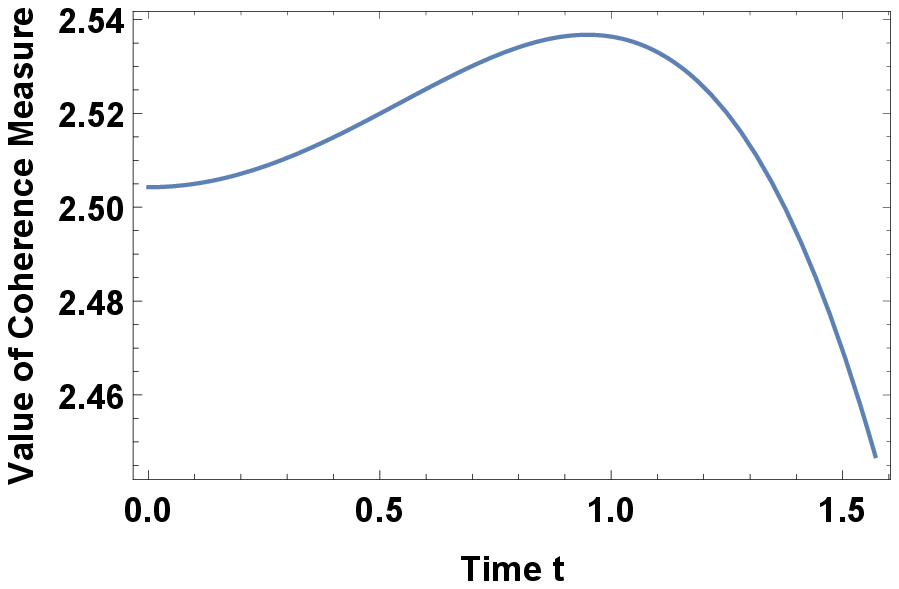} & & \includegraphics[height=3.0cm]{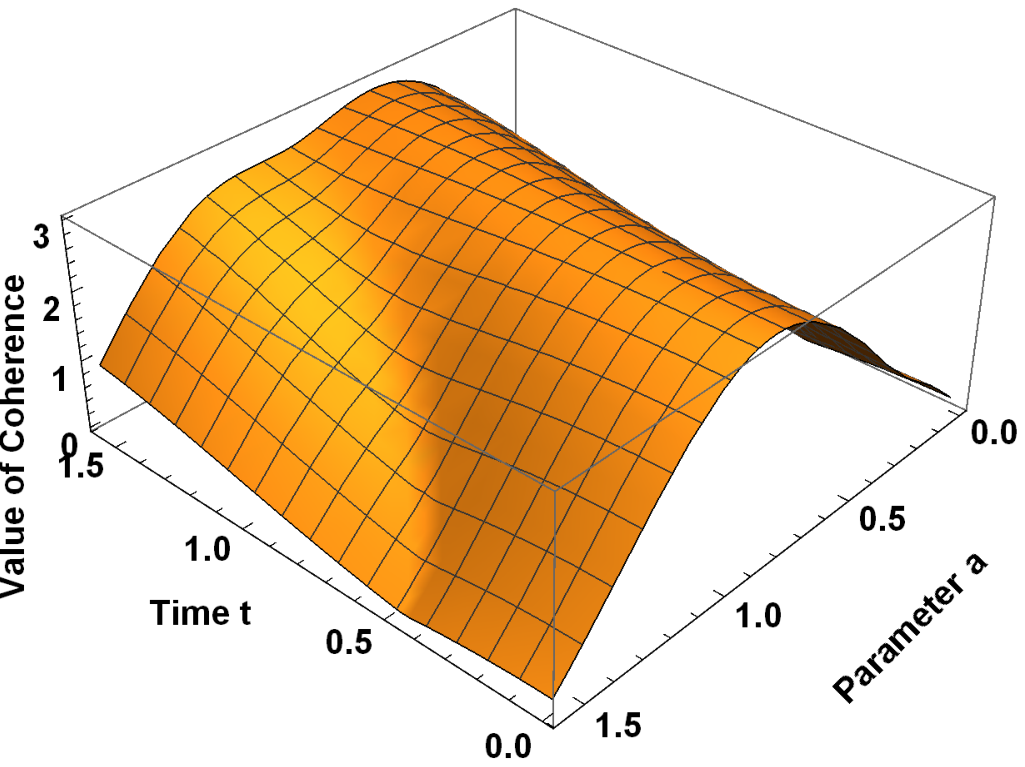}
\end{tabular}
\end{center}
\begin{center}
\begin{tabular}{ccc}
(e) & & (f) \\
\includegraphics[height=3.0cm]{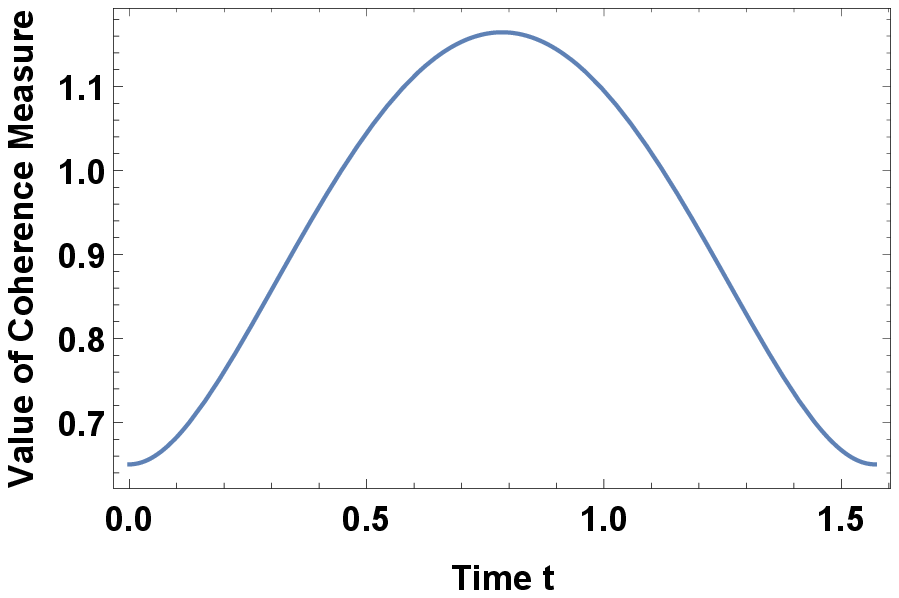} & & \includegraphics[height=3.0cm]{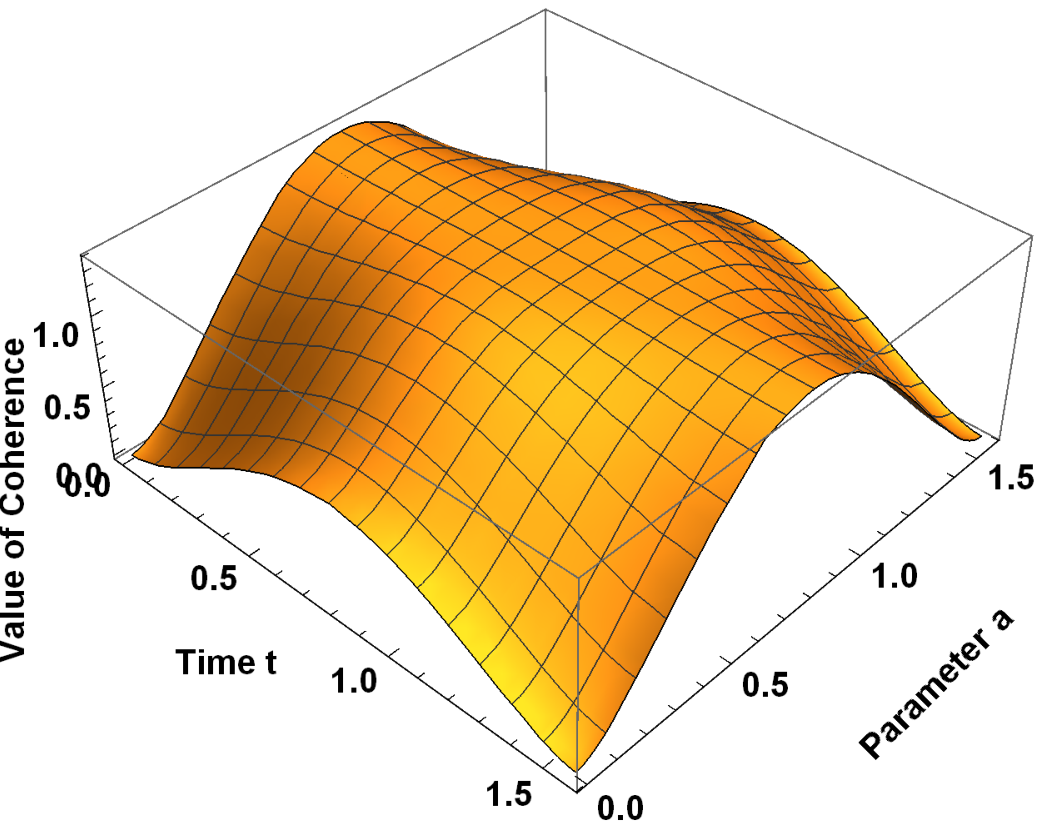} \\
(g) & & (h) \\
\includegraphics[height=3.0cm]{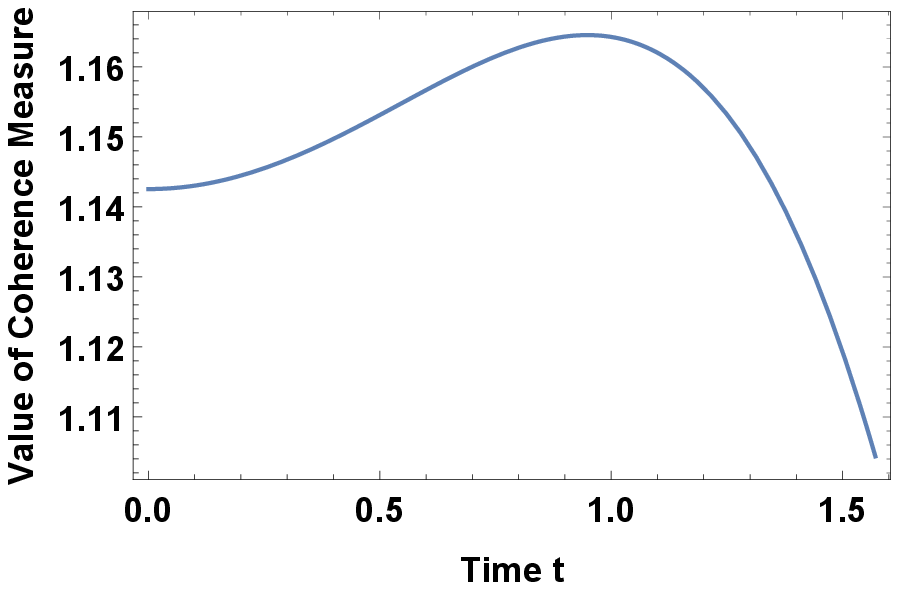} & & \includegraphics[height=3.0cm]{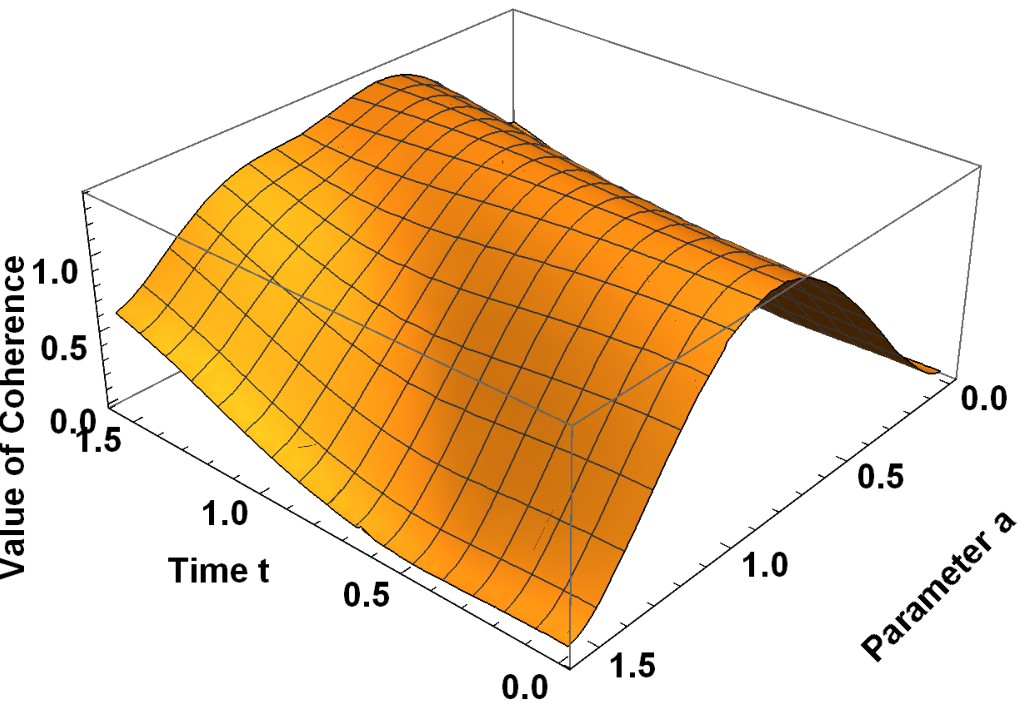}
\end{tabular}
\end{center}
\caption{The changes in value of $C_{l_1}$ measure: (a) for some exemplary states A and B given in (\ref{lbl:eq:given:state:MS:JW:CN:2018}); (b) for a generalized case given in (\ref{lbl:eq:state:for:a:paramter:MS:JW:CN:2018}). Plots (c) and (d) present the changes in value of Coherence measure $C_{l_1}$ when the noise, generated by DM interaction with intensity $D_z=0.5$, is present: (c) for some exemplary states A and B given in (\ref{lbl:eq:given:state:MS:JW:CN:2018}); (d) for a generalized case given in (\ref{lbl:eq:state:for:a:paramter:MS:JW:CN:2018}). The changes in value of relative entropy measure $C_{re}$: (e) for some particular states; (f) for states A and B given in Eq.~(\ref{lbl:eq:state:for:a:paramter:MS:JW:CN:2018}). Plots (g) and (h) show the changes in value of $C_{re}$ measure when the noise, generated by DM interaction with intensity $D_z=0.5$, is present: (g) for states given in~(\ref{lbl:eq:given:state:MS:JW:CN:2018}); (h) for states A and B given in~(\ref{lbl:eq:state:for:a:paramter:MS:JW:CN:2018})
}
\label{lbl:fig:coherence:for:A0:and:AB:states:MS:JW:CN:2018}
\label{lbl:fig:coherence:for:A0:and:AB:states:DM:interaction:MS:JW:CN:2018}
\label{lbl:fig:coherence:relative:entropy:for:A0:and:AB:states:MS:JW:CN:2018}
\label{lbl:fig:coherence:relative:entropy:for:A0:and:AB:states:with:DM:MS:JW:CN:2018}
\end{figure}

The $C_{l_1}$ measure, utilized in quantum systems without presence of noise, offers also the direct possibility of checking if states $A$ and $B$ are the same -- that is if the switch operates on state $\mket{AA1}$.

\begin{theorem}
For the quantum switch defined as operation $U_{qs}$, like in (\ref{lbl:eq:unitary:evolution:for:QS:MS:JW:CN:2018}), the value of $C_{l_1}$ for state $\mket{AA1}$ remains unchanged in time $t$.
\label{lbl:thm:coherence:for:AA:state:MS:JW:CN:2018}
\end{theorem}
\begin{proof}
The above theorem may be proofed by calculating the value of $C_{l_1}$ for the state $\mket{AA1}$. After some transformations of eq.~(\ref{lbl:eq:cocherence:for:AB:MS:JW:CN:2018}) we obtain:
\begin{equation}
C_{l_1}(\rho_{AA}) = 4 |\alpha_0\beta_0| \left( |\alpha_0|^2 + |\alpha_0\beta_0| + |\beta_0|^2 \right) .
\end{equation}
\end{proof}
A consequence described by Thm.~\ref{lbl:thm:coherence:for:AA:state:MS:JW:CN:2018} may be observed at Fig.~\ref{lbl:fig:coherence:for:A0:and:AB:states:MS:JW:CN:2018} in case (b) where the top of the chart is flatten for $a=\pi/4$.

If the noise modelled by DM interaction is present in the system, the $C_{l_1}$ measure shows the distortions' influence on a quantum state:  		
\begin{multline}
C_{l_1}^{\mathrm{DM}}(\rho_{AB}) = 2 \left| \left(\xi  \alpha_1 \beta_0 + \eta  \alpha_0 \beta_1 \right) \left(\zeta  \alpha_1 \beta_0 + \xi \alpha_0 \beta_1\right)\right| \\ 
+ 2 \left| \alpha_0 \alpha_1 \left(\zeta  \alpha_1 \beta_0 + \xi  \alpha_0 \beta_1\right)\right| + 2 \left| \beta_0 \beta_1 \left(\zeta  \alpha_1 \beta_0 + \xi  \alpha_0 \beta_1\right)\right| \\
+ 2 \left| \alpha_0 \alpha_1 \left( \xi  \alpha_1 \beta_0 + \eta \alpha_0 \beta_1\right)\right| +2 \left| \beta_0 \beta_1 \left(\xi \alpha_1 \beta_0+\eta  \alpha_0 \beta_1\right)\right| + 2 \left| \alpha_0 \alpha_1 \beta_0 \beta_1\right|
\label{lbl:eq:cl1:norm:for:dm::MS:JW:CN:2018}
\end{multline}
where the following denotations were used:
\begin{eqnarray}
\xi & = & \frac{1}{2} e^{-\sqrt{-4 D_z^2-t^2}}+\frac{1}{2} e^{\sqrt{-4 D_z^2-t^2}} \notag \\
\eta & = & \frac{e^{\sqrt{-4 D_z^2-t^2}} (-2 D_z - \imag t)}{2 \sqrt{-4 D_z^2-t^2}}-\frac{e^{-\sqrt{-4 D_z^2-t^2}} (-2 D_z - \imag t)}{2 \sqrt{-4 D_z^2-t^2}} \notag  \\
\zeta & = & \frac{e^{\sqrt{-4 D_z^2-t^2}} (2 D_z - \imag t)}{2 \sqrt{-4 D_z^2-t^2}}-\frac{e^{-\sqrt{-4 D_z^2-t^2}} (2 D_z - \imag t)}{2 \sqrt{-4 D_z^2-t^2}} \notag 
\end{eqnarray}
We can observe that time $t$ and the level of noise $D_z$ directly affect a quantum state by modelling its amplitudes. Naturally, these modifications cause changes of the $C_{l_1}$ measure value. 

Utilizing the definitions of states given in (\ref{lbl:eq:given:state:MS:JW:CN:2018}) and (\ref{lbl:eq:state:for:a:paramter:MS:JW:CN:2018}) we present an exemplary process of $C_{l_1}$ value's changes when the noise is generated by DM interaction, what is depicted at Fig.~\ref{lbl:fig:coherence:for:A0:and:AB:states:DM:interaction:MS:JW:CN:2018}. The given characteristic was calculated for $D_z=0.5$. 

The properties (P1)--(P4) indicate that the values of relative entropy of coherence measure $C_{re}$ also correctly depict the differences in a system's operating without and with presence of the DM interaction. For quantum states given in Eq.~(\ref{lbl:eq:given:state:MS:JW:CN:2018}) and Eq.~(\ref{lbl:eq:state:for:a:paramter:MS:JW:CN:2018}) values of relative entropy measure $C_{re}$ are shown at Fig.~\ref{lbl:fig:coherence:relative:entropy:for:A0:and:AB:states:MS:JW:CN:2018}. Additionally, Fig.~\ref{lbl:fig:coherence:relative:entropy:for:A0:and:AB:states:with:DM:MS:JW:CN:2018} demonstrates values of $C_{re}$ measure with distortions generated by DM interaction.

Naturally, the Thm.~\ref{lbl:thm:coherence:for:AA:state:MS:JW:CN:2018} may be also based on relative entropy measure. It may be observed that the value of this measure is:
\begin{multline}
C_{\mathrm{re}} ( \rho_{AA} )= -2 \left(\left| \alpha_0\right|^2 \left| \beta_0\right| ^2 \left(\log \left(\left| \alpha_0\right| ^2 \left| \beta_0\right|^2\right)\right) \right. \\ \left. + 2 \left| \alpha_0\right|^4 (\log  \left| \alpha_0\right| )+2 \left| \beta_0\right| ^4 (\log  \left| \beta_0\right|)\right),
\end{multline}
and again the value of time $t$ is not used.

Regardless the measure, we can observe the difference between values of $C_{l_1}$ in systems with and without noise. If the difference were always equal zero that implies the distortions are impossible to detect.

\begin{proposition}
For the switch's states $\rho_{AB}$ and $\rho^{DM}_{AB}$ we denote the value of difference based on the $C_{l_1}$ measure: 
\begin{multline}
C^{\Delta}_{l_1}( \rho_{AB}, \rho^{DM}_{AB} ) = -2 \left| \zeta  \alpha_1 \beta_0+\xi  \alpha_0 \beta_1\right|  \left(\left| \xi  \alpha_1 \beta_0 + \eta  \alpha_0 \beta_1\right| +\left| \alpha_0 \alpha_1\right| +\left| \beta_0 \beta_1\right| \right) \\ 
- 2 \left(\left| \alpha_0 \alpha_1\right| +\left| \beta_0 \beta_1\right| \right) \left| \xi \alpha_1 \beta_0+\eta \alpha_0 \beta_1\right| + \frac{1}{2} \left| \left(\alpha_1 \beta_0+\alpha_0 \beta_1\right){}^2 \right. \\ 
\left. - e^{4 \imag t} \left(\alpha_1 \beta_0-\alpha_0 \beta_1\right){}^2\right| + 2 \left(\left| \alpha_0 \alpha_1\right| +\left| \beta_0 \beta_1\right| \right) \left(\left| \cos(t) \alpha_0 \beta_1 - \imag \sin(t) \alpha_1 \beta_0\right| \right. \\ \left.
+ \left| \cos(t) \alpha_1 \beta_0 - \imag \sin(t) \alpha_0 \beta_1\right| \right),
\end{multline}
where $\xi$, $\eta$, $\zeta$ are denoted as in Eq.~(\ref{lbl:eq:cl1:norm:for:dm::MS:JW:CN:2018}).
\label{lbl:prop:diff:concurense:measure:MS:JW:CN:2018}
\end{proposition}

Utilizing Prop.~\ref{lbl:prop:diff:concurense:measure:MS:JW:CN:2018} we can directly present an exemplary values of errors occurring during the switch's operating what is depict at Fig.~\ref{lbl:fig:coherence:delta:for:A0:and:AB:states:with:DM:MS:JW:CN:2018}. It should be stressed that the value of difference for exemplary plots at Fig.~\ref{lbl:fig:coherence:delta:for:A0:and:AB:states:with:DM:MS:JW:CN:2018} equals zero only for two points of time t.

An analysis of works \cite{Ekert2002} and \cite{DAriano2005} shows that it is possible to design a quantum circuit for an estimation of Coherence measure value. Fig.~\ref{lbl:fig:coherence:circuit:MS:JW:CN:2018} depicts a circuit built to realize this task. The circuit takes an advantage on quantum states overlapping to estimate a value of Coherence measure -- it is approximated with use of quadratic functional estimation of given density operator.

A significant task is to calculate the values of input states $\mket{\alpha_{p_i}}$ and to denominate the form of gate $\pi_i$, which usually is the SWAP gate. In \cite{Girolami2014} it was presented that using the circuit shown at Fig.~\ref{lbl:fig:coherence:circuit:MS:JW:CN:2018} we can estimate the value of Coherence measure utilizing the Wigner-Yanase-Dyson skew information \cite{Wigner1963}. 

We would like to stress that the estimation of Coherence measure's value for states $\mket{A}$ and $\mket{B}$	may be calculated by a measurement, for example, of state $\mket{\alpha_{p_0}}$, what was shown in \cite{Ekert2002}. The probability that $\mket{\alpha_{p_0}}=\mket{0}$ may be estimated as:
\begin{equation}
\mtr{ \rho_A \rho_B} = \pi_0 = 2 P_0 - 1 
\end{equation}
where $\rho_A = \mket{A}\mbra{A}$, $\rho_B = \mket{B}\mbra{B}$ and $P_0$ is measure projector, whereas $\pi_0$ represents estimated value of coherence. The value of this probability estimates also functionals values, including the values of Coherence measure. An accuracy evaluation of the circuit (Fig.~\ref{lbl:fig:coherence:circuit:MS:JW:CN:2018}) is under the analysis in currently prepared paper \cite{MS:JW:in:preparation:2018}.

\begin{figure}
\begin{center}
\begin{tabular}{ccc}
(a) & &  (c) \\
\includegraphics[width=5.75cm]{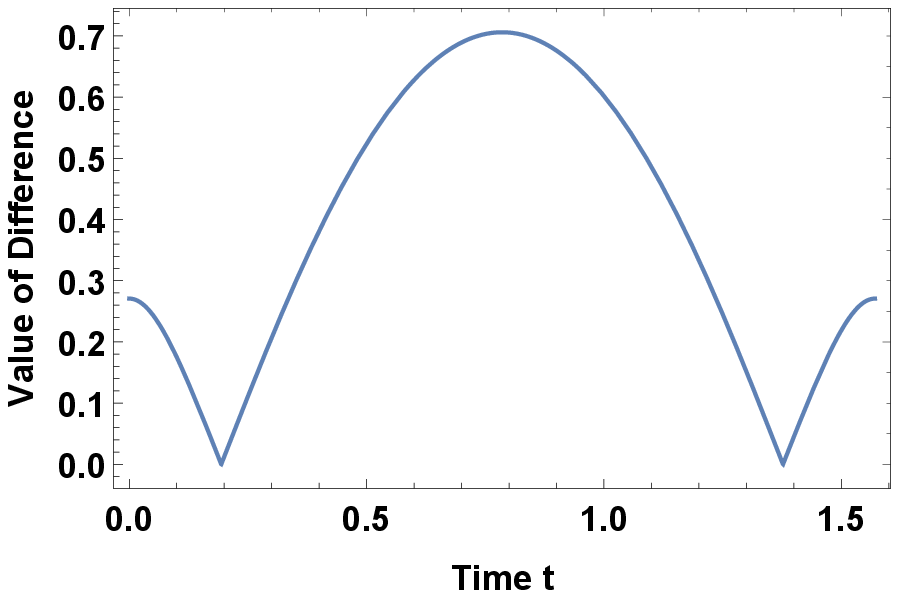} & & \includegraphics[width=5.75cm]{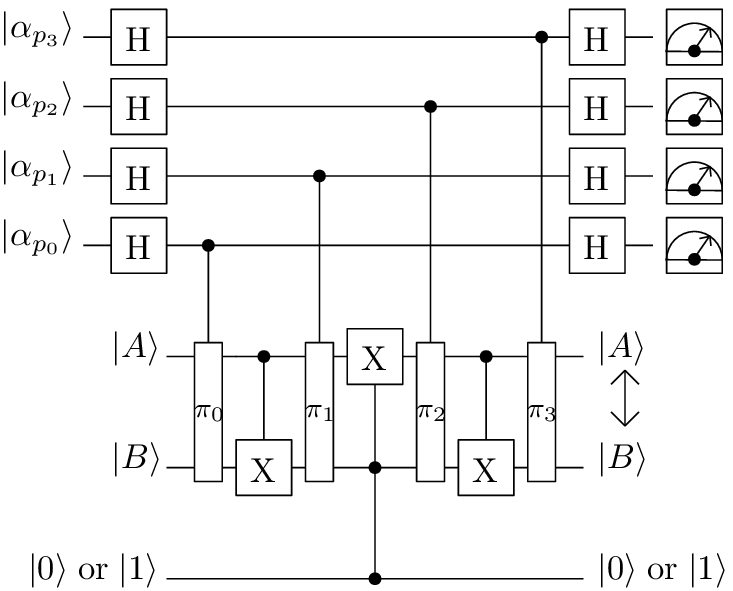} \\
(b) & &  \\
\includegraphics[width=5.75cm]{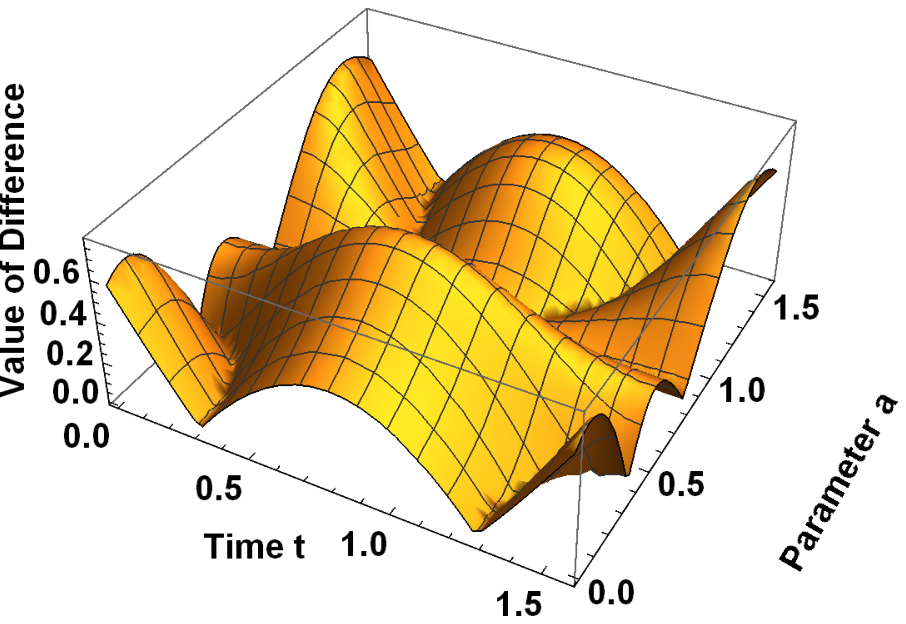} & & {\parbox[b]{5cm}{A quantum circuit (c) for estimation of Coherence measure values during the operating of quantum switch. The values of Coherence measure are calculated for four moments: before the switch starts operating, in two moments of operating and at the end of the process. The controlled $\pi_i$ gate plays a main role in the estimation process of Coherence measure's value}}
\end{tabular}
\end{center}
\caption{The changes in value of absolute difference $C^{\Delta}_{l_1}$: (a) for states given in~(\ref{lbl:eq:given:state:MS:JW:CN:2018}); (b) for states A and B given in~(\ref{lbl:eq:state:for:a:paramter:MS:JW:CN:2018}). The parameter $D_z = 0.5$. Figure (c) represents a quantum circuit for estimation of Coherence measure}
\label{lbl:fig:coherence:delta:for:A0:and:AB:states:with:DM:MS:JW:CN:2018}
\label{lbl:fig:coherence:circuit:MS:JW:CN:2018}
\end{figure}

\section{Conclusions} \label{lbl:sec:conclusions:MS:JW:CN:2018}

In this chapter we discussed utilizing the Coherence measure to trace and evaluate the correctness of quantum switch's operating. First, we described the switch with use of a Hamiltonian to be able to analyze the evolution of quantum system processed by the circuit shown at Fig.~\ref{lbl:fig:QS:EE:MS:JW:CN:2018}. We chose the Dzyaloshinskii-Moriya interaction as a source of noise, which was also modeled as a Hamiltonian. Then the numerical simulations were performed to evaluate a behavior of the circuit. As a tool to capture the differences between the switch operating, simulated with noise or without it, two Coherence measures were used: $l_1$-norm coherence ($C_{l_1}$) and relative entropy of coherence ($C_{\mathrm{re}}$).  

We can observe the difference of $C_{l_1}$ and $C_{\mathrm{re}}$ values if the system works with and without a noise. This shows that the Coherence measures allow to state if the switch operates properly. In addition, an interesting behavior of the system was observed when the switch worked on an initial state $\mket{AA1}$. The experiment showed that the value of $C_{l_1}$ was constant during the simulation -- for other states $\mket{AB1}$, where $A \ne B$, the value of measure evaluates in some characteristic periodic way. This implies that we can also check if the states are the same with use of $C_{l_1}$ measure and quantum switch.

Unfortunately, we also observed that for the systems without and with distortions, modeled as DM interaction, the differences in Coherence measure value are significant even if their intensity $D_z$ is quite low. It is inconvenient if we would like to capture the evolution of the system with the different levels of noise.	

We presented a quantum circuit which estimates the value of Coherence measure. Nowadays, technical solutions based on quantum optics \cite{Abubakar2015}, \cite{Mehic2016}, \cite{Aguado2017} and also physical implementations of qubits allow to build this kind of circuits with use of beam splitters, phase shifters and mirrors \cite{Milburn1989}, \cite{Lloyd1992}.

\subsubsection*{Acknowledgments}

We would like to thank for useful discussions with the~\textit{Q-INFO} group at the Institute of Control and Computation Engineering (ISSI) of the University of Zielona G\'ora, Poland. We would like also to thank to anonymous referees for useful comments on the preliminary version of this chapter. The numerical results were done using the hardware and software available at the ''GPU $\mu$-Lab'' located at the Institute of Control and Computation Engineering of the University of Zielona G\'ora, Poland. 

\label{lbl:references:ms:jw:cn2017}

\end{document}